\documentclass[12pt]{article}
\usepackage{fullpage}
\RequirePackage{amsthm,amsmath,amssymb,amsfonts,mathtools}
\RequirePackage{natbib}
\RequirePackage{graphicx}
\RequirePackage{enumerate}
\RequirePackage{hyperref}
\usepackage{amsmath} 
\DeclareMathOperator{\Cov}{Cov}
\DeclareMathOperator{\diag}{diag}
\numberwithin{equation}{section}
\usepackage{amsmath}

\usepackage{booktabs} 
\usepackage{siunitx}  
\theoremstyle{plain}
\newtheorem{theorem}{Theorem}[section]
\newtheorem{proposition}[theorem]{Proposition}
\newtheorem{corollary}[theorem]{Corollary}
\newtheorem{lemma}[theorem]{Lemma}

\theoremstyle{definition}

\newtheorem{assumption}[theorem]{Assumption}

\theoremstyle{remark}
\newtheorem{remark}[theorem]{Remark}

\newcommand{\R}{\mathbb{R}}
\newcommand{\E}{\mathbb{E}}
\newcommand{\Var}{\mathrm{Var}}
\newcommand{\Pbb}{\mathbb{P}}
\newcommand{\tr}{\mathrm{tr}}
\newcommand{\law}{\mathcal{L}}

\newcommand{\norm}[1]{\left\lVert #1\right\rVert}
\def\X{\boldsymbol{X}}
\def\U{\boldsymbol{U}}
\def\G{\mathbf{\Gamma}}
\def\bms{\mathbf{\Sigma}}
\newcommand{\Fthree}{\mathcal{F}_3}
\newcommand{\distthree}[2]{\left\lVert \law(#1)-\law(#2)\right\rVert_3}
\def\cd{\xrightarrow{d}}
\def\0{\boldsymbol{0}}
\def\1{\boldsymbol{1}}
\providecommand{\thanksnewlabel}[2]{}

\title{Note on High Dimensional Spatial-Sign Test for One Sample Problem}

\author{Ping Zhao \\
School of Mathematical Science, Tiangong University\\
and\\
Long Feng\\
School of Statistics and Data Science, \\
LEBPS, KLMDASR, and LPMC, Nankai University }

\begin{document}

\maketitle

\begin{abstract}
We revisit the null distribution of the high-dimensional spatial-sign test of \cite{wang2015high} under mild structural assumptions on the scatter matrix. We show that the standardized test statistic converges to a non-Gaussian limit, characterized as a mixture of a normal component and a weighted chi-square component. To facilitate practical implementation, we propose a wild bootstrap procedure for computing critical values and establish its asymptotic validity. Numerical experiments demonstrate that the proposed bootstrap test delivers accurate size control across a wide range of dependence settings and dimension–sample-size regimes.


\end{abstract}

\section{Introduction}

Testing mean vectors is a cornerstone problem in multivariate analysis, with ubiquitous roles in genomics, finance, neuroimaging, and many other areas where scientific questions are naturally phrased as location effects.
In the classical low-dimensional regime, Hotelling's $T^2$ test provides an optimal likelihood-based benchmark under multivariate normality \citep{hotelling1931generalization}.
However, modern applications routinely operate in the ``large $p$, small $n$'' or ``large $p$, comparable $n$'' regimes, where the sample covariance matrix is singular or ill-conditioned and Hotelling-type procedures become infeasible or unstable.

A fundamental line of work bypasses covariance inversion by constructing quadratic-form statistics based on the squared $\ell_2$-norm of the sample mean (or mean difference) and suitable bias correction.
Seminal examples include the two-sample statistic of \citet{bai1996effect}, and the U-statistic refinement of \citet{chenqin2010twosample}, which has motivated an extensive literature on high-dimensional mean testing. Please see an overview in \cite{huang2022overview}.

Despite their broad applicability, many asymptotic calibrations in this literature rely on nontrivial restrictions on the covariance eigenstructure (e.g., trace-type conditions that enforce asymptotic Gaussianity of quadratic forms) and typically assume light-tailed distributions with sufficiently high finite moments.
To relax covariance assumptions, several authors have proposed calibration schemes that adapt to non-Gaussian and nonstandard limits, including randomization-based approaches.
In particular, \citet{wangxu2019randomization} studied feasible high-dimensional randomization tests for mean vectors, and \citet{wang2022approximate} developed approximate randomization tests under weak conditions without imposing eigenstructure constraints in related two-sample settings. \noindent
Zhang et al.~\citep{zhang2020simple,zhang2022normalref} developed a \emph{normal-reference} line of tests for high-dimensional mean problems.
In particular, \citet{zhang2020simple} proposed an $L_2$-norm based two-sample procedure that remains valid under mild covariance conditions and calibrates the null distribution via a Welch--Satterthwaite $\chi^2$ approximation, which is designed to accommodate generally non-Gaussian quadratic-form limits. 
Building on the same principle, \citet{zhang2022normalref} provided a unified normal-reference framework for testing high-dimensional mean vectors across several settings and illustrated its practical performance in applications.
More recently, \citet{zhu2024two} studied the high-dimensional two-sample Behrens--Fisher problem and proposed a normal-reference $F$-type test, where the null law is approximated by an $F$-type mixture (a ratio of two independent chi-square-type mixtures) with consistently estimated degrees of freedom for implementation.

Heavy-tailed data and outliers present an additional challenge.
In finance and omics studies, distributional departures from Gaussianity are common, and moment-based procedures can suffer size distortions or power loss.
A robust alternative is to replace raw observations by \emph{spatial signs}, which retain directional information while downweighting large radii.
Under elliptically symmetric models, spatial-sign-based procedures are natural and often stable \citep{oja2010multivariate}.

In this direction, \citet{wang2015high} proposed a high-dimensional nonparametric test for the mean vector based on spatial signs, and \citet{zhou2019adaptiveSS} further developed an adaptive spatial-sign-based approach for elliptically distributed high-dimensional data, together with a chi-square approximation for calibration.
However, the limiting null law of spatial-sign quadratic forms can be non-Gaussian under general covariance structures: the limit may involve a mixture of a Gaussian component and an infinite weighted chi-square component, with weights determined by the (normalized) spectrum of the spatial-sign scatter.
Consequently, fixed-form approximations (normal or chi-square) can be inaccurate in regimes where a few leading eigen-directions are non-negligible.

This paper makes two contributions.
First, we reestablish the null limit distribution of the spatial-sign test statistic of \citet{wang2015high} under mild conditions on the spatial-sign scatter matrix, allowing for general eigenvalue configurations.
We show that the subsequential limits take a mixed form consisting of a Gaussian part and a weighted chi-square part.
Second, we propose a wild bootstrap procedure that consistently estimates critical values under elliptically symmetric distributions, thereby adapting automatically to all possible mixed limits.
We provide a detailed theoretical justification for the bootstrap consistency and demonstrate via simulations that the proposed method controls empirical sizes accurately across a wide range of covariance structures and distributional scenarios.

The remainder of the paper is organized as follows.
Section~\ref{sec:setup} introduces the model, the spatial-sign statistic, and the mixed-limit characterization.
Section~\ref{sec:wildbootstrap} presents the wild bootstrap and establishes its validity.
Section~\ref{sec:simulation} reports numerical results, and Section~\ref{sec:conclusion} concludes with discussions and future directions.

\section{Spatial Sign based Test}\label{sec:setup}
Let $\X_1,\cdots,\X_n$ follows the following elliptical symmetric distribution:
\begin{align}
\X_i=\boldsymbol{\mu}+r_i\G \boldsymbol{u}_i,
\end{align}
where the scatter matrix $\bms=\G\G^\top \in \mathbb{R}^{p\times p}$, $r_i\ge0$ is a scalar radius, $\boldsymbol{u}_i$ is uniform on the unit sphere $\mathbb S^{p-1}$, and $r_i$ is independent of $\boldsymbol{u}_i$. We consider the following hypothesis test problem:
\begin{align}
H_0:\boldsymbol{\mu}=\boldsymbol{0},~~ \text{versus},~~ H_1:\boldsymbol{\mu}\not=\boldsymbol{0}
\end{align}
The spatial sign function is defined as $U(\mathbf{x})=\|\mathbf{x}\|^{-1} \mathbf{x} I(\mathbf{x} \neq \mathbf{0})$. In traditional fixed $p$ circumstance, the following so-called "inner centering and inner standardization" sign-based procedure is usually used (cf., Chapter 6 of \cite{oja2010multivariate})
$$
Q_n^2=n p {\breve{\U}}^T \breve{\boldsymbol{U}},
$$
where $\breve{\boldsymbol{U}}=\frac{1}{n} \sum_{i=1}^n \hat{\boldsymbol{U}}_i, \hat{\boldsymbol{U}}_i=U\left(\mathbf{S}^{-1 / 2} \boldsymbol{X}_{i}\right), \mathbf{S}^{-1 / 2}$ are Tyler's scatter matrix (cf., Section 6.1.3 of \cite{oja2010multivariate}). $Q_n^2$ is affine-invariant and can be regarded as a nonparametric counterpart of Hotelling's $T^2$ test statistic by using the spatial-signs instead of the original observations $\boldsymbol{X}_{i j}$ 's. However, when $p>n, Q_n^2$ is not defined as the matrix $\mathbf{S}^{-1 / 2}$ is not available in high-dimensional settings. So \cite{wang2015high} defined their test statistic
$$
T_{W P L}=\frac{\sum_{i<j} \boldsymbol{U}_i^{\top} \boldsymbol{U}_j}{\sqrt{\frac{n(n-1)}{2}} \widehat{\operatorname{tr}\left(\mathbf{\Sigma}_U^2\right)}},
$$
where
$$
\widehat{\operatorname{tr}\left(\mathbf{\Sigma}_U^2\right)}=\frac{1}{n(n-1)} \operatorname{tr}\left\{\sum_{j \neq k}\left(\boldsymbol{U}_j-\overline{\boldsymbol{U}}_{(j, k)}\right) \boldsymbol{U}_j^{\top}\left(\boldsymbol{U}_k-\overline{\boldsymbol{U}}_{(j, k)}\right) \boldsymbol{U}_k^{\top}\right\}
$$
is an estimator of $\operatorname{tr}\left(\mathbf{\Sigma}_U^2\right)$. Here $\boldsymbol{U}_i=U(\X_i)$ and $\bms_U=E(\U_i\U_i^\top)$. And $\overline{\U}_{(j,k)}=\frac{1}{n-2}\sum_{i\not=j,k}\U_i$. Under some sparity structure assumption of the scatter matrix $\bms$, they establish the asymptotic normality of their test statistic, i.e. $T_{WPL}\cd N(0,1)$. However, when this assumption are not statisfied, as shown in \cite{zhou2019adaptiveSS}, the asympotic normality do not hold anymore and the empirical sizes of the $T_{WPL}$ test is often larger than the nominal level. For any fixed finite $p$, they show that
\begin{align*}
n\|\overline{\U}\|^2 \cd \sum_{r=1}^p \lambda_{p,r}A_r
\end{align*}
where $\lambda_{p,r}$ being the eigenvalues of $\bms_U$ and $A_1,\cdots,A_p$ being i.i.d $\chi_1^2$ random variables.  The above expression also holds for $p= \infty$ provided that $\lim _{p \cd \infty} \bms_U=\bms_{\infty}$ and $\lim _{p \cd \infty} \lambda_{p, r}= \lambda_{\infty, r}$ for all $r=1,2, \ldots$ uniformly where $\lambda_{\infty, r}$ 's are the eigenvalues of $\bms_{\infty}$.
Then, they used $\chi^2$-approximation to approximate the above limited null distribution and reject the null hypothesis if $\chi^2_{\hat{d}}>\hat{d}n\|\overline{\U}\|^2$ and $\hat{d}=n/((n-1)\widehat{\operatorname{tr}\left(\mathbf{\Sigma}_U^2\right)})$. Even the $\chi^2$-approximation is good choice, however, it is not the exact critical value. In the simulation studies, we found that it always has small size than the nonimal level. So, we would recalculate the limit null distribution of $\sum_{i<j} \boldsymbol{U}_i^{\top} \boldsymbol{U}_j$ and proposed two wild bootstrap method to calculate the critical value.

\subsection{Limit null distribution}
 Consider
\[
S_n=\sum_{1\le i<j\le n}\U_i^\top \U_j.
\]
Using $\norm{\U_i}^2=1$, we have the identity
\begin{equation}\label{eq:Sn-identity}
S_n=\frac12\left\|\sum_{i=1}^n \U_i\right\|^2-\frac{n}{2}.
\end{equation}

Set
\begin{align} \label{tn}
\tau:=\tr(\bms_U^2),\qquad
\sigma_n^2:=\binom{n}{2}\tau,\qquad
T_n:=\frac{S_n}{\sigma_n}.
\end{align}

Let $\Fthree$ be the class of functions $f:\R\to\R$ that are three times continuously differentiable and satisfy
\[
\|f^{(r)}\|_\infty\le 1,\qquad r=0,1,2,3.
\]
Define, for real random variables $X,Y$,
\begin{equation}\label{eq:dist3}
\distthree{X}{Y}
:=\sup_{f\in\Fthree}\bigl|\E f(X)-\E f(Y)\bigr|.
\end{equation}
Convergence in $\|\cdot\|_3$ implies weak convergence. Let $G_1,\ldots,G_n$ be i.i.d.\ $N(0,\bms_U)$, independent of $\{\U_i\}$. Define
\[
S_n^{(G)}:=\sum_{1\le i<j\le n}G_i^\top G_j,\qquad
T_n^{(G)}:=\frac{S_n^{(G)}}{\sigma_n}.
\]
Let $\xi\sim N(0,I_p)$ and define the Gaussian quadratic form
\begin{equation}\label{eq:Qp}
Q_p:=\frac{\xi^\top \bms_U \xi-\tr(\bms_U)}{\sqrt{2\tr(\bms_U^2)}}
=\frac{\xi^\top \bms_U \xi-1}{\sqrt{2\tau}}.
\end{equation}

Define
\begin{equation}\label{eq:kappa4}
\kappa_4
:=\frac{\E[(U_1^\top U_2)^4]}{\bigl(\E[(U_1^\top U_2)^2]\bigr)^2}
=\frac{\E[(U_1^\top U_2)^4]}{\tau^2}.
\end{equation}

\begin{theorem}\label{thm:universality}
Assume $\Pbb(X_1=0)=0$ and $\tau=\tr(\bms_U^2)>0$. Then
\begin{equation}\label{eq:univ-bound}
\distthree{T_n}{T_n^{(G)}}
\ \le\ C\,\kappa_4^{3/4}\,n^{-1/2},
\end{equation}
for a universal constant $C>0$. Consequently, if $\kappa_4=o(n^{2/3})$, then
\[
\distthree{T_n}{T_n^{(G)}}\to 0.
\]
\end{theorem}

\begin{theorem}\label{thm:gauss-reduction}
With $T_n^{(G)}$ and $Q_p$ defined in \eqref{eq:Qp},
\begin{equation}\label{eq:gauss-reduction}
T_n^{(G)}=Q_p+o_p(1),
\qquad n\to\infty,
\end{equation}
uniformly over all positive semidefinite $\bms_U$ with $\tr(\bms_U)=1$.
\end{theorem}

Based on Theorems~\ref{thm:universality}--\ref{thm:gauss-reduction}, we can directly have the following corollary.
\begin{corollary}\label{cor:quad-approx}
Under the assumptions of Theorems~\ref{thm:universality}--\ref{thm:gauss-reduction} and $\kappa_4=o(n^{2/3})$,
\[
\distthree{T_n}{Q_p}\to 0.
\]
\end{corollary}
Hence the asymptotic behavior of $T_n$ is determined by $\bms_U$ through the Gaussian quadratic form $Q_p$, with no eigenstructure assumptions on $\bms$.

According to the Corollary 1 in \cite{wang2022approximate} about the limit distribution of the Gaussian quadratic form, we have the following excat distribution of our test statistic $T_n$.

\begin{corollary}\label{cor:all-limits}
Let $\lambda_{1}\ge\cdots\ge\lambda_{p}\ge0$ be eigenvalues of $\bms_U$, and define normalized weights
\[
\alpha_i:=\frac{\lambda_i}{\sqrt{\sum_{j=1}^p\lambda_j^2}}
=\frac{\lambda_i}{\sqrt{\tau}},
\qquad \sum_{i=1}^p\alpha_i^2=1.
\]
Along any subsequence for which $\alpha_i\to \alpha_i^\star$ for each fixed $i$ (allowing infinitely many nonzero limits),
\[
Q_p \cd\left(1-\sum_{i\ge1}(\alpha_i^\star)^2\right)^{1/2}Z_0
+\frac{1}{\sqrt2}\sum_{i\ge1}\alpha_i^\star\,(Z_i^2-1) \doteq T_\infty,
\]
where $Z_0,Z_1,Z_2,\ldots$ are i.i.d.\ $N(0,1)$. The same subsequential limits hold for $T_n$ whenever $\distthree{T_n}{Q_p}\to0$.
\end{corollary}

\begin{corollary}[A sufficient condition for asymptotic normality]\label{cor:clt}
If
\[
\max_{1\le i\le p}\alpha_i \to 0
\quad\Longleftrightarrow\quad
\frac{\tr(\bms_U^4)}{\tr(\bms_U^2)^2}\to 0,
\]
then $Q_p\cd N(0,1)$, and hence $T_n\cd N(0,1)$ under $\distthree{T_n}{Q_p}\to0$.
\end{corollary}

Next, we give some discussion about the parameter $\kappa_4$.

\begin{proposition}[A universal bound]\label{prop:kappa4-universal}
For any symmetric distribution on the unit sphere with $\E U=0$ and $\tau>0$,
\[
1\ \le\ \kappa_4\ \le\ \frac{1}{\tau}\ \le\ p .
\]
\end{proposition}

\begin{proof}
Let $V:=U_1^\top U_2\in[-1,1]$ and $W:=V^2\in[0,1]$. By Jensen's inequality,
$\E(W^2)\ge \E(W)^2$, hence $\kappa_4=\E(W^2)/\E(W)^2\ge 1$.
Moreover, since $0\le W\le 1$ we have $W^2\le W$, so $\E(W^2)\le \E(W)=\tau$ and
$\kappa_4\le \tau/\tau^2=1/\tau$.
Finally, for any positive semidefinite matrix with $\tr(\bms_U)=1$,
$\tr(\bms_U^2)\ge 1/p$ (minimized at $\bms_U=I_p/p$), hence $1/\tau\le p$.
\end{proof}

\begin{remark}[Spherical case]\label{rem:spherical}
If $\bms\propto \mathbf{I}_p$, then $U$ is uniform on $\mathbb S^{p-1}$ and
\[
\E[(\U_1^\top \U_2)^2]=\frac{1}{p},
\qquad
\E[(\U_1^\top \U_2)^4]=\frac{3}{p(p+2)},
\qquad
\kappa_4=\frac{3p}{p+2}\to 3.
\]
\end{remark}

\begin{remark}[When $\kappa_4$ is typically bounded for ACG directions]\label{rem:kappa4-bounded}
For angular central Gaussian (ACG) directions
$\U=\bms^{1/2}S/\sqrt{S^\top \bms S}$ with $S\sim\mathrm{Unif}(\mathbb S^{p-1})$,
a common high-dimensional regime is the effective-rank condition
\[
\frac{\tr(\bms^2)}{\tr(\bms)^2}\to 0
\qquad
\bigl(\text{equivalently } r_{\mathrm{eff}}(\bms):=\tr(\bms)^2/\tr(\bms^2)\to\infty\bigr).
\]
Under this condition, the quadratic form $S^\top\bms S$ concentrates around $\tr(\bms)/p$,
and one can show that for any deterministic unit vector $a\in\mathbb S^{p-1}$,
\[
\frac{\E[(a^\top U)^4]}{\bigl(\E[(a^\top U)^2]\bigr)^2}
=\frac{3p}{p+2}\,(1+o(1))\to 3.
\]
Combined with a mild delocalization condition on $\bms_U$
(e.g., $\tr(\bms_U^4)/\tr(\bms_U^2)^2\to 0$), this yields $\kappa_4\to 3$.
We emphasize that this argument does \emph{not} require a bounded condition number for $\bms$;
it only requires that the spectral energy of $\bms$ is not concentrated in finitely many directions.
\end{remark}

\begin{remark}
Assume the scatter matrix is compound symmetric,
\[
\bms=\sigma^2\{(1-\rho)I_p+\rho \mathbf 1\mathbf 1^\top\},\qquad -1/(p-1)<\rho<1,
\]
and let $u_0=\mathbf 1/\sqrt p$. Write $S=Z u_0+\sqrt{1-Z^2}\,V$ with
$S\sim\mathrm{Unif}(\mathbb S^{p-1})$, where $Z=u_0^\top S$ and
$V\in u_0^\perp\cap\mathbb S^{p-2}$ is uniform and independent of $Z$.
For $U=\bms^{1/2}S/\sqrt{S^\top\bms S}$, define $A=u_0^\top U$ and $B=\sqrt{1-A^2}$.
Then $U=A u_0+B\widetilde V$ with $\widetilde V\in u_0^\perp\cap\mathbb S^{p-2}$ uniform
conditional on $(A,B)$.

Let $T=Z^2\sim \mathrm{Beta}(1/2,(p-1)/2)$ and
\[
\gamma=\frac{\rho p}{1-\rho}.
\]
A direct eigen-decomposition yields
\[
A^2=\frac{(1+\gamma)T}{1+\gamma T}.
\]
Denote $m_2=\E(A^2)$ and $m_4=\E(A^4)$.
For $V:=U_1^\top U_2$ with independent copies $U_1,U_2$, we have the exact identities
\[
\E(V^2)=m_2^2+\frac{(1-m_2)^2}{p-1},
\]
and
\[
\E(V^4)=m_4^2+\frac{6}{p-1}(m_2-m_4)^2+\frac{3}{(p-1)(p+1)}(1-2m_2+m_4)^2.
\]
Therefore,
\[
\kappa_4=\frac{\E(V^4)}{\{\E(V^2)\}^2}
=
\frac{m_4^2+\frac{6}{p-1}(m_2-m_4)^2+\frac{3}{(p-1)(p+1)}(1-2m_2+m_4)^2}
{\left(m_2^2+\frac{(1-m_2)^2}{p-1}\right)^2}.
\]

In the regime $\rho p=O(1)$ (including $\rho=0$ and the admissible negative correlations),
one has $m_2\sim 1/p$ and $m_4\sim 3/(p(p+2))$, which implies $\kappa_4\to 3$.
If $\rho\in(0,1)$ is fixed, then $\gamma\asymp p$ and $pT \cd W\sim\chi_1^2$,
so $A^2 \cd \rho W/((1-\rho)+\rho W)$ and hence $\kappa_4\to K(\rho)\in(1,\infty)$.
In particular, $\kappa_4=O(1)$ for fixed $\rho>0$.
\end{remark}

\section{Wild Bootstrap} \label{sec:wildbootstrap}
Let $\hat{\boldsymbol{\mu}}$ be the sample spatial median, i.e.
\[
\hat{\boldsymbol{\mu}} \in \arg\min_{{\boldsymbol{\mu}}\in\mathbb R^p}\sum_{i=1}^n \|\X_i-{\boldsymbol{\mu}}\|.
\]
Define centered sign vectors $\hat \U_i := U(\X_i-\hat{\boldsymbol{\mu}})$ and the empirical sign-scatter
\[
\hat\bms_U := \frac{1}{n}\sum_{i=1}^n \hat \U_i \hat \U_i^\top.
\]
Let $e_1,\ldots,e_n$ be i.i.d.\ Rademacher variables independent of the data:
$\mathbb P(e_i=1)=\mathbb P(e_i=-1)=1/2$.
Define the bootstrap pseudo-sample $X_i^*:=e_i(\X_i-\hat{\boldsymbol{\mu}})$ and
\[
\hat \U_i^* := U(\X_i^*) = U(e_i(\X_i-\hat{\boldsymbol{\mu}})) = e_i\,\hat \U_i .
\]
Define
\begin{equation}\label{eq:Tn_star_def}
T_R^*
:= \frac{1}{\tau^{1/2}}\cdot
\frac{1}{\sqrt{\binom{n}{2}}}\sum_{1\le i<j\le n} (\hat \U_i^*)^\top \hat \U_j^*
=
\frac{1}{\tau^{1/2}}\cdot
\frac{1}{\sqrt{\binom{n}{2}}}\sum_{1\le i<j\le n} e_ie_j\,\hat \U_i^\top \hat \U_j .
\end{equation}
Note that our bootstrap implementation does not require estimating the scale parameter \(\tau\).
This is because \(T_n\) and \(T_R^*\) share the same scale factor, which cancels out in the comparison
with the bootstrap critical value and therefore does not affect the rejection rule.
We include \(\tau\) in the definition of the test statistics solely for theoretical convenience.

Let $\mathcal L^*(\cdot)$ and $\mathbb P^*(\cdot)$ denote conditional law/probability given the data $\{\X_1,\cdots,\X_n\}$.

To get the consistency of the bootstrap procedure, we need the following assumptions.
\begin{assumption}\label{ass:mom_center}
Define $\zeta_k=E\left(r_i^{-k}\right), r_i=\left\|\boldsymbol{X}_i-\boldsymbol{\mu}\right\|_2, \nu_i=\zeta_1^{-1} r_i^{-1}$. We assume that
$\zeta_k \zeta_1^{-k}<\zeta \in(0, \infty)$ for $k=1,2,3,4$ and all $d$.
\end{assumption}

\begin{assumption}\label{ass:bahadur}
$\limsup _d\|\mathbf{S}\|_2<1-\psi<1$ for some positive constant $\psi$.
\end{assumption}

\begin{assumption}\label{ass:trace_reg}
As $n\to\infty$, the ratio $\kappa_4=\tr(\bms_U^4)/\tr(\bms_U^2)^2$
is bounded: $\sup_n \kappa_4 < \infty$.
\end{assumption}

\begin{lemma}\label{lem:rad_to_gauss}
Condition on the data and define the symmetric matrix $A=(a_{ij})_{1\le i,j\le n}$ by
$a_{ii}=0$ and, for $i\neq j$,
\[
a_{ij}
:= \frac{1}{\tau^{1/2}}\cdot\frac{1}{\sqrt{\binom{n}{2}}}\, \hat \U_i^\top \hat \U_j .
\]
Let $W(e):=\sum_{1\le i<j\le n} a_{ij}e_ie_j$ with $e_i$ Rademacher, and
$W(g):=\sum_{1\le i<j\le n} a_{ij}g_ig_j$ with $g_i\stackrel{iid}{\sim}N(0,1)$ independent.
Then
\[
\big\|\mathcal L^*(W(e))-\mathcal L^*(W(g))\big\|_3
\le C \Big(\sum_{1\le i<j\le n} a_{ij}^4\Big)^{1/4},
\]
for a universal constant $C>0$.
Moreover, 
\[
\Big(\sum_{1\le i<j\le n} a_{ij}^4\Big)^{1/4} = O_p(\tau^{-1/4}n^{-1/2}),
\]
hence $\|\mathcal L^*(W(e))-\mathcal L^*(W(g))\|_3 = O_p(\tau^{-1/4}n^{-1/2})$.
\end{lemma}


\begin{theorem}\label{thm:wild_bootstrap}
Suppose Assumptions~\ref{ass:mom_center}--\ref{ass:trace_reg} hold.
Then, as $n, p\to\infty$,
\[
\big\|\mathcal L^*(T_R^*)-\mathcal L(T_\infty)\big\|_3
\to 0,
\]
if $n\tau\to\infty$.
\end{theorem}

According to Theorem \ref{thm:wild_bootstrap}, we know that the wild bootstrap procedure based on Rademacher variables can consistenty estimate the critical value of $T_\infty$. By Lemma \ref{lem:rad_to_gauss}, we also consider the following wild bootstrap method based on normal variables: we randomly generate $g_i\stackrel{iid}{\sim}N(0,1)$ and then consider the following bootstrap test statistic:
\begin{align}
T^*_N=\frac{1}{\tau^{1/2}}\cdot
\frac{1}{\sqrt{\binom{n}{2}}}\sum_{1\le i<j\le n} g_ig_j\,\hat \U_i^\top \hat \U_j 
\end{align}

Based on the wild bootstrap test statistics \(T_R^*\) and \(T_N^*\), we generate \(M\) bootstrap replicates
\(\{T_{R,m}^*\}_{m=1}^M\) and \(\{T_{N,m}^*\}_{m=1}^M\), respectively, and use their empirical \((1-\alpha)\)-quantiles
as critical values. The null hypothesis is rejected whenever the observed test statistic exceeds the corresponding
bootstrap critical value, i.e.,
\[
T_{n} > \widehat{c}_{R,1-\alpha}
\quad\text{or}\quad
T_{n} > \widehat{c}_{N,1-\alpha},
\]
depending on whether the \textsc{TR} or \textsc{TN} procedure is applied.

\section{Simulation}\label{sec:simulation}
\label{sec:sim_design}
We evaluate the finite-sample performance of four procedures: the test of \citet{wang2015high} (denoted \textsc{WPL}), the adaptive spatial-sign test of \citet{zhou2019adaptiveSS} (denoted \textsc{ZGCZ}), the Rademacher wild bootstrap calibration (denoted \textsc{TR}), and the Gaussian wild bootstrap calibration (denoted \textsc{TN}).
For each replication, we generate i.i.d.\ observations $X_1,\ldots,X_n\in\mathbb{R}^p$
from an elliptically symmetric model with equicorrelated dependence.
Specifically, we set $\boldsymbol\mu=\0_p$ and
\begin{equation}\label{eq:Sigma_equicorr}
\bms = (1-\rho) \mathbf{I}_p + \rho\, \1_p \1_p^{\top},
\end{equation}
so that $\Sigma_{jj}=1$ and $\Sigma_{jk}=\rho$ for $j\neq k$.
We consider $\rho\in\{0.1,0.5,0.9\}$, sample sizes $n\in\{40,80,120\}$, and dimensions
$p\in\{100,200,400\}$.
For each $(n,p,\rho)$ configuration, we perform $10{,}000$ replications.

We examine three representative regimes.
\begin{enumerate}
\item \textit{Normal distribution} $X_i \sim N_p(\0_p,\bms)$.
\item \textit{multivariate t-distribution} $X_i \sim t_{\nu}(\0_p,\bms)$ with $\nu=3$ degrees of freedom.
\item \textit{Mixture normal distribution} $\X_i \sim 0.8 N(\0_p, \bms)+0.2 N(\0_p,9\bms)$
\end{enumerate}

All procedures are implemented at nominal level $\alpha=0.05$.
For bootstrap-based calibrations, we use $B=500$ bootstrap resamples within each Monte Carlo replication.
We report empirical rejection probabilities across replications for each $(n,p,\rho)$ setting. To assess performance of a test in maintaining the nominal size (type I error), we use the following so-called average relative error $\operatorname{ARE}=100 M^{-1} \sum_{j=1}^M\left|\hat{\alpha}_j-\alpha\right| / \alpha$, where $\hat{\alpha}_j, j=1, \ldots, M$ denote the empirical sizes under consideration. A smaller
ARE value indicates a better overall performance of the associated test in terms of size control.

For power comparison, we set $\boldsymbol{\mu}=(\delta,\cdots,\delta)$ where $\delta=2\sqrt{\tr^{1/2}(\bms^2)/(np)}$ where $\bms=\Cov(\X)$ is the covariance matrix of $\X$.

Tables~\ref{tab1}--\ref{tab6} summarize the empirical sizes and powers of the four competing procedures (WPL, ZGCZ, TR and TN) across three dependence levels ($\rho=0.1,0.5,0.9$), three sample sizes ($n=40,80,120$), and three dimensions ($p=100,200,400$) under multivariate normal, multivariate $t$, and a mixture normal model.
Overall, the proposed wild bootstrap calibrations TR and TN exhibit stable and accurate size control: their empirical sizes are consistently close to the nominal 5\% level in all combinations of $(n,p,\rho)$ and under all three distributions, with no visible deterioration as the correlation strengthens or as the dimension increases.
In contrast, WPL tends to be liberal, with empirical sizes typically in the 6\%--7\% range and becoming more inflated as $\rho$ increases, while ZGCZ is severely conservative, producing near-zero rejection rates for $\rho=0.1$ and remaining far below 5\% even for $\rho=0.9$.
This size distortion is also reflected by the reported ARE values in Tables~\ref{tab1}--\ref{tab3}, where TR/TN are close to the benchmark while ZGCZ exhibits extremely large ARE, indicating substantial mismatch between its reference approximation and the true finite-sample null behavior.

The power comparisons in Tables~\ref{tab4}--\ref{tab6} further reinforce these findings.
Under the multivariate normal model (Table~\ref{tab4}), TR and TN attain powers comparable to WPL, whereas ZGCZ has dramatically lower power, consistent with its conservative calibration.
Under heavy-tailed $t$ and mixture normal distributions (Tables~\ref{tab5} and \ref{tab6}), all spatial-sign-based procedures achieve high power, but the advantage of the bootstrap calibration becomes most pronounced in terms of reliability: TR and TN maintain near-nominal size while delivering strong power across all dependence levels, including the highly correlated case $\rho=0.9$.
Across the three distributions, TR and TN perform similarly, with TR being marginally more stable in a few settings, suggesting that the choice between Rademacher and Gaussian multipliers is not critical in practice.

Taken together, these results indicate that the proposed wild bootstrap provides a robust and broadly applicable calibration for high-dimensional spatial-sign testing, delivering accurate size control and competitive power under both light-tailed and heavy-tailed models and under general covariance dependence.

\begin{table}[htbp]
\centering
\setlength{\tabcolsep}{3pt} 
\caption{Empirical sizes (\%) of four methods for multivariate normal distribution.}
\begin{tabular}{cc|cccc|cccc|cccc}
\toprule
&&\multicolumn{4}{c}{$\rho=0.1$}&\multicolumn{4}{c}{$\rho=0.5$}&\multicolumn{4}{c}{$\rho=0.9$}\\ \hline
$n$ & $p$ & WPL &ZGCZ &TR &TN & WPL &ZGCZ &TR &TN & WPL &ZGCZ &TR &TN \\
\midrule
40 & 100 & 6.04 & 0.01 & 5.14 & 5.41 & 7.32 & 1.03 & 5.81 & 6.04 & 6.97 & 2.20 & 5.49 & 5.70 \\
80 & 100 & 6.37 & 0.00 & 5.49 & 5.62 & 6.93 & 1.22 & 5.59 & 5.87 & 6.93 & 2.17 & 5.41 & 5.36 \\
120 & 100 & 5.90 & 0.02 & 5.17 & 5.05 & 6.20 & 1.03 & 4.89 & 4.79 & 7.03 & 2.35 & 5.44 & 5.64 \\ \hline
40 & 200 & 6.05 & 0.02 & 5.09 & 5.23 & 6.94 & 1.05 & 5.31 & 5.61 & 6.98 & 2.23 & 5.41 & 5.79 \\
80 & 200 & 6.03 & 0.01 & 5.11 & 5.17 & 6.84 & 1.01 & 5.34 & 5.44 & 7.41 & 2.32 & 5.71 & 5.78 \\
120 & 200 & 6.61 & 0.01 & 5.89 & 5.87 & 6.81 & 1.06 & 5.27 & 5.46 & 6.72 & 2.14 & 5.32 & 5.42 \\ \hline
40 & 400 & 6.49 & 0.01 & 5.22 & 5.39 & 7.25 & 0.93 & 5.68 & 5.79 & 7.13 & 2.17 & 5.69 & 5.80 \\
80 & 400 & 6.59 & 0.00 & 5.36 & 5.46 & 7.08 & 1.09 & 5.58 & 5.78 & 6.67 & 2.18 & 5.35 & 5.41 \\
120 & 400 & 6.16 & 0.00 & 5.03 & 5.26 & 6.75 & 1.02 & 4.97 & 5.14 & 7.11 & 2.37 & 5.53 & 5.47 \\
\midrule
\text{ARE (\%)} & & 22.76 & 195.11 & 12.36 & 13.04 & 30.38 & 89.20 & 16.40 & 16.84 & 36.31 & 56.09 & 13.16 & 15.58 \\
\bottomrule
\end{tabular}
\label{tab1}
\end{table}

\begin{table}[htbp]
\centering
\setlength{\tabcolsep}{3pt}
\caption{Empirical sizes (\%) of four methods for multivariate $t$-distribution.}
\begin{tabular}{cc|cccc|cccc|cccc}
\toprule
&&\multicolumn{4}{c}{$\rho=0.1$}&\multicolumn{4}{c}{$\rho=0.5$}&\multicolumn{4}{c}{$\rho=0.9$}\\ \hline
$n$ & $p$ & WPL &ZGCZ &TR &TN & WPL &ZGCZ &TR &TN & WPL &ZGCZ &TR &TN \\
\midrule
40 & 100 & 6.17 & 0.00 & 5.19 & 5.46 & 7.80 & 1.09 & 6.09 & 6.37 & 7.22 & 2.46 & 5.74 & 6.08 \\
80 & 100 & 5.95 & 0.03 & 5.08 & 5.21 & 6.86 & 1.03 & 5.25 & 5.40 & 6.85 & 2.21 & 5.24 & 5.37 \\
120 & 100 & 6.24 & 0.02 & 5.22 & 5.45 & 6.58 & 1.04 & 5.17 & 5.21 & 7.20 & 2.39 & 5.43 & 5.62 \\ \hline
40 & 200 & 6.60 & 0.00 & 5.47 & 5.82 & 7.47 & 0.99 & 5.86 & 6.04 & 7.21 & 2.23 & 5.59 & 5.90 \\
80 & 200 & 5.79 & 0.01 & 4.88 & 4.85 & 7.15 & 1.11 & 5.38 & 5.57 & 7.10 & 2.35 & 5.46 & 5.58 \\
120 & 200 & 6.33 & 0.00 & 5.47 & 5.37 & 6.99 & 1.02 & 5.18 & 5.29 & 6.93 & 2.09 & 5.33 & 5.33 \\ \hline
40 & 400 & 6.83 & 0.01 & 5.58 & 5.90 & 6.99 & 1.01 & 5.49 & 5.64 & 6.94 & 2.14 & 5.38 & 5.59 \\
80 & 400 & 6.38 & 0.01 & 5.09 & 5.38 & 7.11 & 1.05 & 5.25 & 5.56 & 7.29 & 2.33 & 5.87 & 5.94 \\
120 & 400 & 6.53 & 0.02 & 5.32 & 5.45 & 6.98 & 1.06 & 5.54 & 5.47 & 6.93 & 2.09 & 5.09 & 5.35 \\
\midrule
\text{ARE (\%)} & & 26.27 & 197.29 & 12.89 & 17.02 & 68.31 & 178.91 & 17.42 & 21.89 & 70.93 & 113.64 & 16.93 & 22.60 \\
\bottomrule
\end{tabular}
\label{tab2}
\end{table}

\begin{table}[htbp]
\centering
\setlength{\tabcolsep}{3pt}
\caption{Empirical sizes (\%) of four methods for mixture normal distribution.}
\begin{tabular}{cc|cccc|cccc|cccc}
\toprule
&&\multicolumn{4}{c}{$\rho=0.1$}&\multicolumn{4}{c}{$\rho=0.5$}&\multicolumn{4}{c}{$\rho=0.9$}\\ \hline
$n$ & $p$ & WPL &ZGCZ &TR &TN & WPL &ZGCZ &TR &TN & WPL &ZGCZ &TR &TN \\
\midrule
40 & 100 & 5.92 & 0.01 & 4.90 & 5.28 & 7.20 & 0.87 & 5.54 & 5.69 & 7.11 & 2.13 & 5.75 & 5.81 \\
80 & 100 & 6.30 & 0.02 & 5.31 & 5.52 & 7.02 & 1.17 & 5.48 & 5.58 & 6.95 & 1.94 & 5.10 & 5.32 \\
120 & 100 & 5.83 & 0.01 & 4.94 & 5.02 & 6.87 & 0.95 & 5.38 & 5.52 & 6.89 & 2.14 & 5.42 & 5.32 \\ \hline
40 & 200 & 6.55 & 0.02 & 5.42 & 5.68 & 6.82 & 1.14 & 5.39 & 5.54 & 7.12 & 2.42 & 5.68 & 5.76 \\
80 & 200 & 6.01 & 0.00 & 5.06 & 5.15 & 6.64 & 1.04 & 5.13 & 5.36 & 6.98 & 2.30 & 5.57 & 5.70 \\
120 & 200 & 6.33 & 0.01 & 5.39 & 5.43 & 6.67 & 0.98 & 5.14 & 5.22 & 7.42 & 2.35 & 5.76 & 5.86 \\ \hline
40 & 400 & 6.60 & 0.00 & 5.26 & 5.58 & 7.33 & 1.10 & 6.01 & 6.09 & 7.43 & 2.19 & 5.74 & 6.10 \\
80 & 400 & 6.72 & 0.01 & 5.48 & 5.41 & 6.73 & 0.96 & 5.22 & 5.28 & 6.73 & 2.17 & 5.15 & 5.35 \\
120 & 400 & 6.60 & 0.01 & 5.29 & 5.43 & 6.77 & 0.96 & 5.33 & 5.30 & 6.75 & 2.13 & 5.08 & 5.12 \\
\midrule
\text{ARE (\%)} & & 24.07 & 196.30 & 11.60 & 14.51 & 33.02 & 85.49 & 15.56 & 18.20 & 42.13 & 59.42 & 15.31 & 17.58 \\
\bottomrule
\end{tabular}
\label{tab3}
\end{table}

\begin{table}[htbp]
\centering
\caption{Empirical power (\%) of four methods for multivariate normal distribution.}
\begin{tabular}{cc|cccc|cccc|cccc}
\toprule
&&\multicolumn{4}{c}{$\rho=0.1$}&\multicolumn{4}{c}{$\rho=0.5$}&\multicolumn{4}{c}{$\rho=0.9$}\\ \hline
$n$ & $p$ & WPL &ZGCZ &TR &TN & WPL &ZGCZ &TR &TN & WPL &ZGCZ &TR &TN \\ \hline
40&100&57.5&3.1&55.2&55.7&54.4&26.1&48.5&50.8&50.4&29.9&44.5&44.3\\
80&100&57.9&5.7&56.4&56.5&54.2&24.9&49.2&49.8&51.8&32.2&46.9&46.7\\
120&100&56.8&5.8&55&55.5&54.8&26.7&49.5&49.2&46.9&31.6&42.4&42.9\\ \hline
40&200&60.5&3.5&58.1&58.7&51.9&23.0&47.4&48.6&48.3&28.9&42.4&43.8\\
80&200&57.1&4.2&53.9&54.7&54.0&26.9&48&49.5&48.0&29.1&42.0&41.6\\
120&200&58.4&4.2&55.7&55.4&51.3&24.6&46.4&46.9&50.5&32.9&45.8&45.7\\ \hline
40&400&54.0&2.7&50.7&51.4&54.0&26.2&49&49.5&48.7&30.1&43.9&44.5\\
80&400&56.3&3.3&53.1&52.7&53.9&24.7&48.9&49.9&46.7&29.7&41.5&41.6\\
120&400&57.7&3.0&53.6&54.5&55.1&26.5&50&50.4&50.6&31.4&45.1&45.9\\
\bottomrule
\end{tabular}
\label{tab4}
\end{table}

\begin{table}[htbp]
\centering
\caption{Empirical power (\%) of four methods for multivariate $t$-distribution.}
\begin{tabular}{cc|cccc|cccc|cccc}
\toprule
&&\multicolumn{4}{c}{$\rho=0.1$}&\multicolumn{4}{c}{$\rho=0.5$}&\multicolumn{4}{c}{$\rho=0.9$}\\ \hline
$n$ & $p$ & WPL &ZGCZ &TR &TN & WPL &ZGCZ &TR &TN & WPL &ZGCZ &TR &TN \\
\midrule
40&100&91.8&36.8&90.9&91.6&86.0&66.2&83.5&83.5&80.9&65.3&76.8&77\\
80&100&92.1&35.6&90.8&91.4&87.3&67&84.9&84.7&84.8&69.8&80.7&80.8\\
120&100&93.4&38.8&92.8&92.6&88.2&67.5&84.7&84.3&84.2&69.8&80.7&80.8\\ \hline
40&200&89.2&27.8&88.8&89.1&85.7&65.5&83.5&83.8&83.5&69.8&79.6&80.4\\
80&200&92.3&30.5&90.1&91.1&88.4&67.9&85.8&85.9&84.1&69.9&80.2&79.9\\
120&200&91.2&33.3&90.6&90.7&87.1&66.2&85.4&85.4&85.4&71.0&81.9&82.5\\ \hline
40&400&88.4&25.9&87.5&87.3&85.4&65.1&81.9&82.6&82.4&66.9&78.7&79.1\\
80&400&91.6&26.7&90.7&89.9&87.5&66.1&85.1&84.9&80.7&67.5&77.3&77.3\\
120&400&92.2&25.6&90.7&90.5&88.3&69.4&85.2&85.4&85.2&72.5&82.5&82.6\\
\bottomrule
\end{tabular}
\label{tab5}
\end{table}

\begin{table}[htbp]
\centering
\caption{Empirical power (\%) of four methods for mixture normal distribution.}
\begin{tabular}{cc|cccc|cccc|cccc}
\toprule
&&\multicolumn{4}{c}{$\rho=0.1$}&\multicolumn{4}{c}{$\rho=0.5$}&\multicolumn{4}{c}{$\rho=0.9$}\\ \hline
$n$ & $p$ & WPL &ZGCZ &TR &TN & WPL &ZGCZ &TR &TN & WPL &ZGCZ &TR &TN \\
\midrule
40&100&93.0&38.5&92.4&91.8&87.3&65.7&84.1&84.7&84.6&71.4&81.5&82.2\\
80&100&92.5&38.9&91.3&91.6&87.5&67.0&84.4&84.8&83.1&70.6&79.5&79.8\\
120&100&92.1&42.1&91.7&91.4&91.6&71.8&88.7&89.1&86.6&71.7&82.5&82.8\\ \hline
40&200&91.5&29.4&89.7&89.7&88.6&68.1&86&86.3&85.2&70.2&81.0&81.4\\
80&200&92.0&33.8&91.4&91.5&90.5&69.5&87.6&87.1&86.3&71.1&81.3&82.1\\
120&200&91.7&34.4&90.1&90.1&89.5&70.2&86.4&86.6&84.3&69.9&80.5&80.6\\ \hline
40&400&90.7&27.4&88.9&89.2&90.1&68.1&87.0&86.5&85.9&74.4&83.4&84.1\\
80&400&92.0&29.4&90.9&91.0&89.9&69.7&87.1&87.6&85.0&71.3&81.4&81.6\\
120&400&92.1&29.6&89.9&91.1&89.5&67.6&86.4&86.4&82.4&69.9&79.5&79.7\\
\bottomrule
\end{tabular}
\label{tab6}
\end{table}
\section{Conclusion}\label{sec:conclusion}
This paper clarifies the null behavior of the high-dimensional spatial-sign test of \cite{wang2015high} under mild assumptions on the dependence structure, showing that the standardized statistic may converge to a non-Gaussian limit that mixes a normal component with a weighted chi-square component. To address the resulting practical challenge of critical value calibration, we develop a wild bootstrap procedure and prove its consistency. Simulation results further indicate that the proposed approach achieves reliable empirical size control across a broad range of correlation patterns and dimension–sample-size configurations.

Several directions appear particularly promising for future research. First, it would be valuable to extend the theory beyond the i.i.d. setting to dependent observations, including time series and spatial data, where robust sign-based procedures are attractive but the limiting behavior can be more intricate. Second, developing power-optimal or adaptive variants—capable of automatically interpolating between dense and sparse alternatives—could improve sensitivity while retaining robustness. Third, it would be useful to incorporate data-driven estimation of the mixing weights (or low-dimensional spectral summaries of the sign-scatter) to yield analytic critical values and reduce bootstrap cost in ultra-high dimensions. Fourth, extending the framework to multi-sample problems (two-sample location, MANOVA-type settings) and to other sign-based statistics would broaden applicability. Finally, a systematic investigation of robustness under contamination, heteroskedasticity, and model misspecification—together with principled tuning and diagnostic tools—would strengthen practical guidance for real data applications.

\section{Appendix}

\subsection{Proof of Theorem~\ref{thm:universality}}

Define the standardized kernel
\[
w_{ij}(x,y):=\frac{x^\top y}{\sigma_n},\qquad 1\le i<j\le n,
\]
and the functional
\[
W(x_1,\ldots,x_n):=\sum_{i<j} w_{ij}(x_i,x_j).
\]
Then $T_n=W(U_1,\ldots,U_n)$ and $T_n^{(G)}=W(G_1,\ldots,G_n)$.

Degeneracy holds because $\E \U_i=\E G_i=0$: for any fixed $a\in\R^p$,
\[
\E[w_{ij}(\U_i,a)]=\sigma_n^{-1}\E(\U_i)^\top a=0,\qquad
\E[w_{ij}(G_i,a)]=0.
\]

Moreover,
\[
\sigma_{ij}^2:=\E[w_{ij}(\U_i,\U_j)^2]
=\frac{\E[(U_1^\top U_2)^2]}{\sigma_n^2}
=\frac{\tau}{\binom{n}{2}\tau}
=\frac{2}{n(n-1)}.
\]
Define the influence of index $k$ by
\[
\mathrm{Inf}_k:=\sum_{j<k}\sigma_{jk}^2+\sum_{j>k}\sigma_{kj}^2
=(n-1)\cdot \frac{2}{n(n-1)}=\frac{2}{n}.
\]
Hence
\begin{equation}\label{eq:inf-sum}
\sum_{k=1}^n \mathrm{Inf}_k^{3/2}
= n\left(\frac{2}{n}\right)^{3/2}
=2^{3/2}\,n^{-1/2}.
\end{equation}

For $k=1,\ldots,n+1$, define
\[
W_k:=W(G_1,\ldots,G_{k-1},U_k,\ldots,U_n).
\]
Then $W_1=T_n$ and $W_{n+1}=T_n^{(G)}$, so for any $f\in\Fthree$,
\[
\E f(W_1)-\E f(W_{n+1})
=\sum_{k=1}^n\Bigl(\E f(W_k)-\E f(W_{k+1})\Bigr).
\]

Fix $k$. Let $W_{k,0}$ denote the same functional with the $k$th input set to $0$:
\[
W_{k,0}:=W(G_1,\ldots,G_{k-1},0,U_{k+1},\ldots,U_n),
\]
where $w_{ij}(0,\cdot)=w_{ij}(\cdot,0)=0$. Then
\[
W_k=W_{k,0}+\Delta_k,\qquad
W_{k+1}=W_{k,0}+\Delta_k',
\]
with
\[
\Delta_k
=\sum_{i<k} w_{ik}(G_i,U_k)+\sum_{j>k} w_{kj}(U_k,\U_j),
\qquad
\Delta_k'
=\sum_{i<k} w_{ik}(G_i,G_k)+\sum_{j>k} w_{kj}(G_k,\U_j).
\]

Apply Taylor's theorem around $W_{k,0}$:
\[
f(W_{k,0}+\Delta)
=f(W_{k,0})+f'(W_{k,0})\Delta+\frac12 f''(W_{k,0})\Delta^2+R,
\]
where the remainder satisfies
\[
|R|\le \frac{1}{6}\|f^{(3)}\|_\infty\,|\Delta|^3\le \frac{|\Delta|^3}{6}.
\]
Therefore,
\begin{align*}
\E\bigl[f(W_k)-f(W_{k+1})\bigr]
&=\E\Bigl[f'(W_{k,0})(\Delta_k-\Delta_k')
+\frac12 f''(W_{k,0})(\Delta_k^2-(\Delta_k')^2)\Bigr]
+\E(R_k-R_k').
\end{align*}
The first-order term vanishes because, conditioning on all variables except $U_k$ (respectively $G_k$), each summand in $\Delta_k$ and $\Delta_k'$ has conditional mean $0$ by degeneracy. The second-order term vanishes because $G_k$ is chosen with matching covariance $\E(G_kG_k^\top)=\E(U_kU_k^\top)=\bms_U$, so the conditional second moments of $\Delta_k$ and $\Delta_k'$ coincide.

Hence,
\begin{equation}\label{eq:one-step}
\bigl|\E f(W_k)-\E f(W_{k+1})\bigr|
\le \frac{1}{6}\E|\Delta_k|^3+\frac{1}{6}\E|\Delta_k'|^3.
\end{equation}

Using $\E|\Delta|^3\le (\E\Delta^4)^{3/4}$, it suffices to bound $\E\Delta_k^4$ and $\E(\Delta_k')^4$.
Conditional on $U_k$, the summands in $\Delta_k$ are independent across indices (they depend on disjoint collections of $G_i$ and $\U_j$), and have conditional mean $0$. A standard fourth-moment inequality for sums of independent mean-zero variables yields
\[
\E(\Delta_k^4)\le C_0\Bigl(\mathrm{Inf}_k^2+\sum_{j\neq k}\E[w_{kj}(\cdot)^4]\Bigr)
\]
for a universal constant $C_0>0$, where the fourth moments are controlled by $\kappa_4$ through
\[
\E[w_{kj}(U_k,\U_j)^4]=\frac{\E[(U_1^\top U_2)^4]}{\sigma_n^4}
=\kappa_4\,\frac{\tau^2}{\sigma_n^4}
=\kappa_4\,\sigma_{kj}^4.
\]
The Gaussian/mixed terms satisfy analogous hypercontractive bounds, so overall
\[
\E(\Delta_k^4)\le C_1\,\kappa_4\,\mathrm{Inf}_k^2
\]
for a universal constant $C_1>0$, and similarly for $\Delta_k'$. Therefore,
\[
\E|\Delta_k|^3\le (\E\Delta_k^4)^{3/4}\le C_1^{3/4}\kappa_4^{3/4}\mathrm{Inf}_k^{3/2},
\qquad
\E|\Delta_k'|^3\le C_1^{3/4}\kappa_4^{3/4}\mathrm{Inf}_k^{3/2}.
\]
Plugging into \eqref{eq:one-step}, summing over $k$, and using \eqref{eq:inf-sum} yields
\[
\sup_{f\in\Fthree}\bigl|\E f(T_n)-\E f(T_n^{(G)})\bigr|
\le C\,\kappa_4^{3/4}\sum_{k=1}^n \mathrm{Inf}_k^{3/2}
\le C\,\kappa_4^{3/4}\,n^{-1/2},
\]
which is \eqref{eq:univ-bound}. \qed

\subsection{Proof of Theorem~\ref{thm:gauss-reduction}}
Using the analogue of \eqref{eq:Sn-identity},
\begin{equation}\label{eq:SnG-identity}
S_n^{(G)}=\frac12\left\|\sum_{i=1}^n G_i\right\|^2-\frac12\sum_{i=1}^n \|G_i\|^2.
\end{equation}
Let $W_n:=\sum_{i=1}^n G_i$. Then $W_n\sim N(0,n\bms_U)$, so $W_n\stackrel{d}{=}\sqrt{n}\,\bms_U^{1/2}\xi$ and
\begin{equation}\label{eq:Wn}
\frac{1}{2}\|W_n\|^2
=\frac{n}{2}\,\xi^\top \bms_U\xi.
\end{equation}
Also, $\E\|G_i\|^2=\tr(\bms_U)=1$. Write
\[
\sum_{i=1}^n\|G_i\|^2 = n + R_n,\qquad \E R_n=0.
\]
Since $\Var(\|G_1\|^2)=2\tr(\bms_U^2)=2\tau$ for Gaussian $G_1\sim N(0,\bms_U)$,
\[
\Var(R_n)=n\Var(\|G_1\|^2)=2n\tau.
\]
Combining with \eqref{eq:SnG-identity}--\eqref{eq:Wn} gives
\[
S_n^{(G)}
=\frac{n}{2}\bigl(\xi^\top \bms_U\xi-1\bigr)-\frac12 R_n.
\]
Divide by $\sigma_n=\sqrt{\binom{n}{2}\tau}\sim n\sqrt{\tau/2}$:
\[
T_n^{(G)}
=\frac{\xi^\top \bms_U\xi-1}{\sqrt{2\tau}}-\frac{R_n}{2\sigma_n}
=Q_p-\frac{R_n}{2\sigma_n}.
\]
Finally,
\[
\Var\!\left(\frac{R_n}{2\sigma_n}\right)
=\frac{2n\tau}{4\sigma_n^2}
=\frac{n\tau}{2\binom{n}{2}\tau}
=\frac{1}{n-1}\to 0,
\]
so $R_n/(2\sigma_n)=o_p(1)$, proving \eqref{eq:gauss-reduction}. \qed

We first restate the Lemma S2.4 in \cite{zhao2024spatial} here.
\begin{lemma}\label{lem:theta_rate}
Under Assumptions~\ref{ass:mom_center}--\ref{ass:bahadur}, let $\hat{\boldsymbol{{\boldsymbol{\theta}}}}=\hat{\boldsymbol{\mu}}-\boldsymbol{\mu}$, we have
\[
\|\hat{\boldsymbol{\theta}}\|_2 = O_p(\zeta_1^{-1}n^{-1/2}).
\]
and 
\begin{equation}\label{eq:bahadur}
\hat{\boldsymbol{\theta}}
=
\zeta_1^{-1}\cdot \frac{1}{n}\sum_{i=1}^n \U_i
+\zeta_1^{-1}\varrho,
\qquad \|\varrho\|_2 = O_p(n^{-1}),
\end{equation}
\end{lemma}

Define
\[
\bms_U := \E(U_iU_i^\top),\qquad 
\tilde\bms_U := \frac1n\sum_{i=1}^n U_iU_i^\top,\qquad
\hat\bms_U := \frac1n\sum_{i=1}^n \hat \U_i\hat \U_i^\top.
\]
\begin{lemma}
\label{lem:Sigma_hat_detailed}
Under Assumptions~\ref{ass:mom_center}--\ref{ass:bahadur}, we have
\begin{equation}\label{eq:trSigma2_rate}
\tr(\hat\bms_U^2)/\tr(\bms_U^2)=1+O_p(n^{-1/2}\tau^{-1/2}),
\end{equation}
\end{lemma}

\begin{proof}
By the triangle inequality,
\begin{equation}\label{eq:tri_decomp}
\|\hat\bms_U-\bms_U\|_F
\le
\|\hat\bms_U-\tilde\bms_U\|_F
+
\|\tilde\bms_U-\bms_U\|_F.
\end{equation}
Note that
\[
\tilde\bms_U-\bms_U
=
\frac1n\sum_{i=1}^n (U_iU_i^\top-\bms_U),
\]
where $\{U_iU_i^\top-\bms_U\}_{i=1}^n$ are i.i.d.\ with mean zero.
Using $\E\|\cdot\|_F^2$ and independence,
\[
\E\|\tilde\bms_U-\bms_U\|_F^2
=
\frac1n\,\E\|U_1U_1^\top-\bms_U\|_F^2.
\]
We now bound $\E\|U_1U_1^\top-\bms_U\|_F^2$.
Since $\|U_1\|_2\le 1$ and $U_1U_1^\top$ is rank one,
\[
\|U_1U_1^\top\|_F^2
=
\tr(U_1U_1^\top U_1U_1^\top)
=
\tr(U_1U_1^\top)
=
\|U_1\|_2^2
\le 1.
\]
Moreover,
\[
\E\,\tr(U_1U_1^\top\bms_U)
=
\tr\big(\E(U_1U_1^\top)\bms_U\big)
=
\tr(\bms_U^2)
=
\|\bms_U\|_F^2.
\]
Therefore,
\begin{align*}
\E\|U_1U_1^\top-\bms_U\|_F^2
&=
\E\|U_1U_1^\top\|_F^2+\|\bms_U\|_F^2
-2\E\,\tr(U_1U_1^\top\bms_U) \\
&\le 1+\|\bms_U\|_F^2-2\|\bms_U\|_F^2
= 1-\|\bms_U\|_F^2
\le 1.
\end{align*}
Consequently,
\[
\E\|\tilde\bms_U-\bms_U\|_F^2 \le \frac1n,
\]
and by Markov's inequality,
\begin{equation}\label{eq:sampling_rate}
\|\tilde\bms_U-\bms_U\|_F = O_p(n^{-1/2}).
\end{equation}

Write
\[
\hat\bms_U-\tilde\bms_U
=
\frac1n\sum_{i=1}^n\big(\hat \U_i\hat \U_i^\top-U_iU_i^\top\big).
\]

Beause
\begin{equation}\label{eq:centering_reduce}
 \|\hat \U_i-\U_i\|_2\le \frac{C\|\hat{\boldsymbol{\theta}}\|_2}{r_i},
\end{equation}
using the triangle inequality and the bound
$\|uu^\top-vv^\top\|_F\le 2\|u-v\|_2$ for any vectors $u,v$ with $\|u\|_2,\|v\|_2\le 1$,
we have
\begin{equation}\label{eq:centering_rate_pre}
\|\hat\bms_U-\tilde\bms_U\|_F
\le
\frac{2}{n}\sum_{i=1}^n \|\hat \U_i-\U_i\|_2
\le
\frac{C\|\hat{\boldsymbol{\theta}}\|_2}{n}\sum_{i=1}^n \frac{1}{r_i}
=
C\|\hat{\boldsymbol{\theta}}\|_2\cdot \Big(\frac1n\sum_{i=1}^n r_i^{-1}\Big).
\end{equation}
Assumption~\ref{ass:mom_center} ensures that $\E\|Y_1\|_2^{-1}$ exists and is of order $\zeta_1$, so by the law of large numbers,
\begin{equation}\label{eq:inv_norm_lln}
\frac1n\sum_{i=1}^n r_i^{-1}
=
O_p(\zeta_1).
\end{equation}
Therefore \eqref{eq:centering_rate_pre} and \eqref{eq:inv_norm_lln} yield
\begin{equation}\label{eq:centering_rate}
\|\hat\bms_U-\tilde\bms_U\|_F
=
O_p\!\big(\zeta_1\|\hat{\boldsymbol{\theta}}\|_2\big).
\end{equation}
Now apply Lemma~\ref{lem:theta_rate}:
\[
\|\hat{\boldsymbol{\theta}}\|_2 = O_p(\zeta_1^{-1}n^{-1/2}),
\]
which implies
\[
\zeta_1\|\hat{\boldsymbol{\theta}}\|_2 = O_p(n^{-1/2}).
\]
Combining this with \eqref{eq:sampling_rate}, \eqref{eq:tri_decomp}, and \eqref{eq:centering_rate}
gives
\[
\|\hat\bms_U-\bms_U\|_F = O_p(n^{-1/2}).
\]

We use the identity
\[
\tr(\hat\bms_U^2)-\tr(\bms_U^2)
=
\tr\big((\hat\bms_U-\bms_U)(\hat\bms_U+\bms_U)\big).
\]
By Cauchy--Schwarz for the Frobenius inner product,
\begin{align*}
\big|\tr(\hat\bms_U^2)-\tr(\bms_U^2)\big|
&\le
\|\hat\bms_U-\bms_U\|_F\cdot \|\hat\bms_U+\bms_U\|_F\\
&\le  \|\hat\bms_U-\bms_U\|_F(\|\hat\bms_U-\bms_U\|_F+2\|\bms_U\|_F)
\end{align*}
So
\begin{align*}
\frac{\tr(\hat\bms_U^2)}{\tr(\bms_U^2)}-1=O_p(n^{-1/2}\tau^{-1/2}).
\end{align*}
\end{proof}

\subsection{Proof of Lemma \ref{lem:rad_to_gauss}}
\begin{proof}
We use a Lindeberg replacement.
Let $(\xi_1,\ldots,\xi_n)$ be i.i.d.\ Rademacher and $(\eta_1,\ldots,\eta_n)$ be i.i.d.\ $N(0,1)$,
independent of each other and of the data.
For $k=0,1,\ldots,n$, define the hybrid vector
\[
Z^{(k)} := (\eta_1,\ldots,\eta_k,\xi_{k+1},\ldots,\xi_n),
\qquad
W_k := \sum_{i<j} a_{ij} Z^{(k)}_i Z^{(k)}_j .
\]
Then $W_0=W(e)$ and $W_n=W(g)$, and by telescoping,
\[
\Big|\mathbb E^* f(W_0)-\mathbb E^* f(W_n)\Big|
\le \sum_{k=1}^n \Big|\mathbb E^*\{f(W_{k-1})-f(W_k)\}\Big|,
\qquad f\in\mathcal F_3 .
\]
Fix $k$. Write $W_{k-1}=\alpha + \xi_k \beta$ and $W_k=\alpha + \eta_k \beta$, where
$\alpha$ and $\beta$ are functions of $\{Z^{(k)}_\ell:\ell\neq k\}$ and the data.
By Taylor expansion of $f(\alpha+t\beta)$ around $t=0$ up to third order,
and using $\mathbb E(\xi_k)=\mathbb E(\eta_k)=0$, $\mathbb E(\xi_k^2)=\mathbb E(\eta_k^2)=1$,
and $\mathbb E(\xi_k^3)=\mathbb E(\eta_k^3)=0$, we obtain
\[
\Big|\mathbb E^*\{f(W_{k-1})-f(W_k)\}\Big|
\le C\, \mathbb E^*(|\beta|^3)\cdot
\Big|\mathbb E(\xi_k^4)-\mathbb E(\eta_k^4)\Big|
\le C'\, \mathbb E^*(|\beta|^3),
\]
since $\mathbb E(\xi_k^4)=1$ and $\mathbb E(\eta_k^4)=3$.
Now $\beta = \sum_{j\neq k} a_{kj} Z^{(k)}_j$.
By Burkholder/Rosenthal and $\|a_{k\cdot}\|_2^2=\sum_{j\neq k}a_{kj}^2$,
$\mathbb E^*(|\beta|^3)\le C \Big(\sum_{j\neq k} a_{kj}^2\Big)^{3/2}$.
Summing over $k$ yields
\[
\big|\mathbb E^* f(W(e))-\mathbb E^* f(W(g))\big|
\le C \sum_{k=1}^n \Big(\sum_{j\neq k} a_{kj}^2\Big)^{3/2}.
\]
Finally, by Hölder,
$\sum_{k=1}^n (\sum_{j\neq k}a_{kj}^2)^{3/2}
\le (\sum_{k=1}^n \sum_{j\neq k} a_{kj}^2)^{3/4} (\sum_{k=1}^n \sum_{j\neq k} a_{kj}^4)^{1/4}
\asymp (\sum_{i<j} a_{ij}^4)^{1/4}$,
because $\sum_{i<j} a_{ij}^2=1$ by construction.
Taking the supremum over $f\in\mathcal F_3$ gives the first bound.

For the rate, note $|a_{ij}|\le C\tau^{-1/2}/\sqrt{\binom{n}{2}}\le C'n^{-1}\tau^{-1/2}$ and $\sum_{i<j}a_{ij}^2$ is bounded,
hence $\sum_{i<j}a_{ij}^4 \le (\max_{i<j}a_{ij}^2)\sum_{i<j}a_{ij}^2 = O_p(\tau^{-1}n^{-2})$,
so $(\sum_{i<j}a_{ij}^4)^{1/4}=O_p(\tau^{-1/4}n^{-1/2})$.
\end{proof}

\subsection{Proof of Theorem \ref{thm:wild_bootstrap}}

\begin{proof}
Let $\hat W=(\hat w_{ij})_{1\le i,j\le n}$ with $\hat w_{ij}=\hat \U_i^\top \hat \U_j$.
Since $\hat g_{ii}=\|\hat \U_i\|^2=1$,
\[
\sum_{i<j} e_ie_j\hat \U_i^\top \hat \U_j
=\frac12\sum_{i\ne j} e_ie_j\hat w_{ij}
=\frac12\,\mathbf e^\top(\hat W-I_n)\mathbf e,
\]
where $\mathbf e=(e_1,\ldots,e_n)^\top$.
Hence
\[
T_R^*=\mathbf e^\top \hat A\,\mathbf e,
\qquad
\hat A:=\frac{1}{2\sqrt\tau\,\sqrt{\binom{n}{2}}}\,(\hat W-I_n),
\qquad \tr(\hat A)=0.
\]

Let $\mathbf g\sim N(0,I_n)$ independent of the data and define
\[
T_N^*:=\frac{1}{\sqrt\tau}\cdot
\frac{1}{\sqrt{\binom{n}{2}}}\sum_{i<j} g_ig_j\,\hat \U_i^\top \hat \U_j
=\mathbf g^\top \hat A\,\mathbf g.
\]
By Lemma \ref{lem:rad_to_gauss}, the conditional laws of $T_R^*$ and $T_N^*$ are asymptotically equivalent. Therefore, it suffices to prove the claimed limit for $T_N^*$.

Let $\hat\nu_1,\ldots,\hat\nu_n$ be eigenvalues of $\hat A$ and write
$\hat A=Q^\top\diag(\hat\nu_1,\ldots,\hat\nu_n)Q$ for some orthogonal $Q$.
Then, conditionally on the data,
\[
T_N^*=\mathbf g^\top \hat A\,\mathbf g
=\sum_{k=1}^n \hat\nu_k(\xi_k^2-1),
\qquad \xi_k\stackrel{i.i.d.}{\sim}N(0,1),
\]
because $\tr(\hat A)=\sum_k \hat\nu_k=0$.
Moreover,
\[
\Var^*(T_N^*)=2\sum_{k=1}^n \hat\nu_k^2=2\tr(\hat A^2).
\]
Using $\hat A=\{2\sqrt\tau\sqrt{\binom{n}{2}}\}^{-1}(\hat W-I_n)$,
we obtain
\[
2\tr(\hat A^2)
=\frac{\tr(\hat\bms_U^2)}{\tau}
=1+o_{\mathbb P}(1),
\]
where the last step follows Lemma \ref{lem:Sigma_hat_detailed}.

Define normalized coefficients
\[
\hat\beta_k:=\frac{\hat\nu_k}{\sqrt{2\tr(\hat A^2)}},
\qquad \sum_{k=1}^n\hat\beta_k^2=1,
\]
and write
\[
\frac{T_N^*}{\sqrt{\Var^*(T_N^*)}}
=\sum_{k=1}^n \hat\beta_k(\xi_k^2-1).
\]

Taking the same procedure as Lemma S.8 and in \cite{wang2022approximate}, along any subsequence such that $\alpha_i\to\alpha_i^\star$ for each fixed $i$,
one may further select a (non-decreasing) integer sequence $r_n^{\*}\to\infty$
so that (i) $\max_{k>r_n^{\*}}|\hat\beta_k|\to 0$ and
(ii) $\sum_{k=1}^{r_n^{\*}}\hat\beta_k^2\to \sum_{i\ge1}(\alpha_i^\star)^2$.

Then for any fixed $m$,
\[
\sum_{k=1}^m \hat\beta_k(\xi_k^2-1)\;\cd\;
\frac{1}{\sqrt2}\sum_{k=1}^m \alpha_k^\star(\xi_k^2-1),
\]
and the tail part satisfies, by Lindeberg--Feller CLT (since $\max_{k>r_n^{\*}}|\hat\beta_k|\to 0$),
\[
\sum_{k=r_n^{\*}+1}^{n}\hat\beta_k(\xi_k^2-1)
\;\cd\;
\left(1-\sum_{i\ge1}(\alpha_i^\star)^2\right)^{1/2}Z_0,
\qquad Z_0\sim N(0,1),
\]
independent of $(\xi_1,\xi_2,\ldots)$.
Letting $m\to\infty$ yields
\[
\frac{T_N^*}{\sqrt{\Var^*(T_N^*)}}
\cd
\left(1-\sum_{i\ge1}(\alpha_i^\star)^2\right)^{1/2}Z_0
+\frac{1}{\sqrt2}\sum_{i\ge1}\alpha_i^\star(Z_i^2-1).
\]
Finally, since $\Var^*(T_N^*)\to 1$ in probability, Slutsky's theorem gives
$T_N^*\cd T_\infty$ conditionally in probability.

By Lemma \ref{lem:rad_to_gauss}, the conditional laws of $T_R^*$ and $T_N^*$ are asymptotically equivalent,
hence $T_R^*\cd T_\infty$ conditionally in probability.
This completes the proof.
\end{proof}

\bibliographystyle{chicago}
\bibliography{ref}

@article{zhou2019adaptiveSS,
  title   = {An adaptive spatial-sign-based test for mean vectors of elliptically distributed high-dimensional data},
  author  = {Zhou, Bu and Guo, Jia and Chen, Jianwei and Zhang, Jin-Ting},
  journal = {Statistics and Its Interface},
  volume  = {12},
  year    = {2019},
  pages   = {93--106}
}

@article{zhao2024spatial,
  title={Spatial sign based principal component analysis for high dimensional data},
  author={Zhao, Ping and Wang, Hongfei and Feng, Long},
  journal={arXiv preprint arXiv:2409.13267},
  year={2024}
}

@article{wang2022approximate,
  title={An approximate randomization test for the high-dimensional two-sample Behrens--Fisher problem under arbitrary covariances},
  author={Wang, Rui and Xu, Wangli},
  journal={Biometrika},
  volume={109},
  number={4},
  pages={1117--1132},
  year={2022},
  publisher={Oxford University Press}
}

@article{hotelling1931generalization,
  author  = {Hotelling, Harold},
  title   = {The Generalization of Student's Ratio},
  journal = {The Annals of Mathematical Statistics},
  year    = {1931},
  volume  = {2},
  number  = {3},
  pages   = {360--378}
}

@article{bai1996effect,
  author  = {Bai, Z. D. and Saranadasa, H.},
  title   = {Effect of High Dimension: By an Example of a Two Sample Problem},
  journal = {Statistica Sinica},
  year    = {1996},
  volume  = {6},
  number  = {2},
  pages   = {311--329}
}

@article{chenqin2010twosample,
  author  = {Chen, Song Xi and Qin, Ying-Lu},
  title   = {A Two-Sample Test for High-Dimensional Data with Applications to Gene-Set Testing},
  journal = {The Annals of Statistics},
  year    = {2010},
  volume  = {38},
  number  = {2},
  pages   = {808--835}
}

@article{wangxu2019randomization,
  author  = {Wang, Rui and Xu, Xingzhong},
  title   = {A feasible high dimensional randomization test for the mean vector},
  journal = {Journal of Statistical Planning and Inference},
  year    = {2019},
  volume  = {199},
  pages   = {160--178}
}

@article{zhang2022normalref,
  author  = {Zhang, Jin-Ting and Zhou, Bu and Guo, Jia},
  title   = {Testing high-dimensional mean vector with applications: {A} normal reference approach},
  journal = {Statistical Papers},
  year    = {2022},
  volume  = {63},
  number  = {2},
  pages   = {363--391}
}

@article{huang2022overview,
  title={An overview of tests on high-dimensional means},
  author={Huang, Yuan and Li, Changcheng and Li, Runze and Yang, Songshan},
  journal={Journal of multivariate analysis},
  volume={188},
  pages={104813},
  year={2022},
  publisher={Elsevier}
}

@article{zhang2020simple,
  title={A simple two-sample test in high dimensions based on L 2-norm},
  author={Zhang, Jin-Ting and Guo, Jia and Zhou, Bu and Cheng, Ming-Yen},
  journal={Journal of the American Statistical Association},
  volume={115},
  number={530},
  pages={1011--1027},
  year={2020},
  publisher={Taylor \& Francis}
}

@article{zhu2024two,
  title={Two-sample Behrens--Fisher problems for high-dimensional data: a normal reference F-type test},
  author={Zhu, Tianming and Wang, Pengfei and Zhang, Jin-Ting},
  journal={Computational Statistics},
  volume={39},
  number={6},
  pages={3207--3230},
  year={2024},
  publisher={Springer}
}

@article{wang2015high,
  title={A high-dimensional nonparametric multivariate test for mean vector},
  author={Wang, Lan and Peng, Bo and Li, Runze},
  journal={Journal of the American Statistical Association},
  volume={110},
  number={512},
  pages={1658--1669},
  year={2015},
  publisher={Taylor \& Francis}
}

@book{oja2010multivariate,
  title={Multivariate nonparametric methods with R: an approach based on spatial signs and ranks},
  author={Oja, Hannu},
  year={2010},
  publisher={Springer Science \& Business Media}
}

@Manual{R,
    title = {R: A Language and Environment for Statistical Computing},
    author = {{R Core Team}},
    organization = {R Foundation for Statistical Computing},
    address = {Vienna, Austria},
    year = {2022},
  }

\end{document}